\documentclass[10pt, conference, letterpaper]{IEEEtran}
\IEEEoverridecommandlockouts
\usepackage{cite}
\usepackage{amsmath,amssymb,amsfonts}
\usepackage{algorithmic}
\usepackage{hyperref}
\usepackage{graphicx}
\usepackage{textcomp}
\usepackage{xcolor}
\usepackage{bm}
\usepackage{amsthm}
\usepackage{tabularx}
\theoremstyle{definition}
\newtheorem{definition}{Definition}
\newtheorem{lemma}{Lemma}
\newcolumntype{P}[1]{>{\centering\arraybackslash}p{#1}}
\newcolumntype{M}[1]{>{\centering\arraybackslash}m{#1}}

\newtheorem{proposition}{Proposition}
\newtheorem{theorem}{Theorem}
\newtheorem{corollary}{Corollary}[lemma]
\usepackage{cite}
\usepackage{amsmath}
\usepackage{amssymb}
\usepackage{amsmath}
\DeclareUnicodeCharacter{2212}{-}
\usepackage{cancel}
\usepackage{algorithm2e}
\RestyleAlgo{ruled}
\usepackage{setspace}

\ifCLASSOPTIONcompsoc
\usepackage[caption=false,font=normalsize,labelfon
t=sf,textfont=sf]{subfig}
\else
\usepackage[caption=false,font=footnotesize]{subfi
g}
\fi
\def\BibTeX{{\rm B\kern-.05em{\sc i\kern-.025em b}\kern-.08em
    T\kern-.1667em\lower.7ex\hbox{E}\kern-.125emX}}
\begin{document}

\title{TCP Slice: A semi-distributed TCP algorithm for Delay-constrained Applications\\
}

\author{\IEEEauthorblockN{Dibbendu Roy}
\IEEEauthorblockA{\textit{Department of Electrical Engineering} \\
\textit{IIT Indore}\\
droy@iiti.ac.in}
\and
\IEEEauthorblockN{Goutam Das}
\IEEEauthorblockA{\textit{G.S. Sanyal School of Telecommunications} \\
\textit{IIT Kharagpur}\\
Kharagpur, India \\
gdas@gssst.iitkgp.ac.in}
% \and
% \and
% \and
% \and
}

\maketitle

\begin{abstract}
The TCP congestion control protocol serves as the cornerstone of reliable internet communication. However, as new applications require more specific guarantees regarding data rate and delay, network management must adapt. Thus, service providers are shifting from decentralized to centralized control of the network using a software-defined network controller (SDN). The SDN classifies applications and  allocates logically separate resources called “slices”, over the physical network. We propose “TCP Slice”, a congestion control algorithm that meets specific delay and bandwidth guarantees. Obtaining closed-form delay bounds for a client is challenging due to dependencies on other clients and their traffic stochasticity. We use network calculus to derive the client's delay bound and incorporate it as a constraint in the Network Utility Maximization problem. We solve the resulting optimization using dual decomposition and obtain a semi-distributed TCP protocol that can be implemented with the help of SDN controller and the use of Explicit Congestion Notification (ECN) bit.  Additionally, we also propose a  proactive approach for congestion control using digital twin. TCP Slice represents a significant step towards accommodating evolving internet traffic patterns and the need for better network management in the face of increasing application diversity.

\end{abstract}

\section{Introduction}
The TCP congestion control protocol is certainly the backbone of reliable internet communication and has been the de facto standard ever since its inception in the 1980s \cite{jacobson1988congestion}. The success of TCP can be attributed to its capability of being deployed in a distributed manner without necessitating complete knowledge of the network state. Thus, each client can send data over TCP without being concerned about the network's other clients, leading to a highly scalable system. However, over the years, the nature of internet traffic has changed significantly due to increased accessibility and the introduction of new applications and services \cite{han2020framework}. The days of being satisfied with basic data delivery assurances are long gone. As an example, Table \ref{tab:reqmap} shows the diverse quality of service (QoS) requirements of 5G use cases \cite{5GKPIs5G30:online}. It is anticipated that this diversity will continue to expand with the emergence of newer application scenarios, and future networks, including 6G \cite{giordani2020toward,6GConne15:online}, and therefore the future internet will be required to accommodate these diverse demands.

\begin{table}[h]
\centering
\caption{Application, Slice and Requirement Mapping in 5G \cite{5GKPIs5G30:online}}
    \label{tab:reqmap}
    
    \begin{tabular}{|M{0.35\columnwidth}|M{0.14\columnwidth}|M{0.35\columnwidth}|}
    \hline
         \bfseries Use Cases & \bfseries Slice & \bfseries QoS \\
         \hline
          Augmented and Virtual reality & eMBB & $\text{Downstream} \geq 100 \text{Mbps}, \text{Uplink} \geq 50 \text{Mbps}, \text{Latency} \leq 4 \text{ms}$\\
         \hline
         Autonomous Cars, Remote Surgery & uRLLC & $\text{Latency} \leq 1 \text{ms}, \text{Reliability} = 99.999\%$\\
         \hline
         IoT, Smart Factories, Platooning & mMTC & $\text{Connection Density}=1\times 10^6 \text{devices/Km}^2$\\
         \hline
    \end{tabular}
\end{table}
To address the aforementioned needs, better management of networks is necessary. Thus, a more centralized control of the network through software-defined networks (SDN) has gained popularity \cite{benzekki2016software,mckeown2008openflow} (see \figurename \ref{fig:sysmod}). Specifically, to meet the various QoS requirements of emerging new applications, the notion of network slicing is introduced \cite{foukas2017network}. The idea behind network slicing is to divide a single physical network infrastructure into multiple virtual networks, each tailored to serve different use cases with distinct performance, security, and resource allocation requirements. Each network slice operates independently and can be managed and optimized separately. Evidently, the implementation of slices requires global network state information, 
making SDN a key enabler for the same. The SDN can run network operations such as routing, network address translation, and network slice management by implementing network functions on top of its control plane (known as Network function virtualization - NFV, see \figurename \ref{fig:sysmod}). As shown in Table \ref{tab:reqmap}, based on QoS requirements, 5G defines three main categories of slices, enhanced Mobile Broadband (eMBB), ultra Reliable Low Latency Communications (uRLLC), and massive Machine Type Communications (mMTC). Keeping these developments in mind, it is important to look into the possibilities of extending the capabilities of traditional TCP algorithms. 
\begin{figure}
\centering\includegraphics[width=0.9\columnwidth]{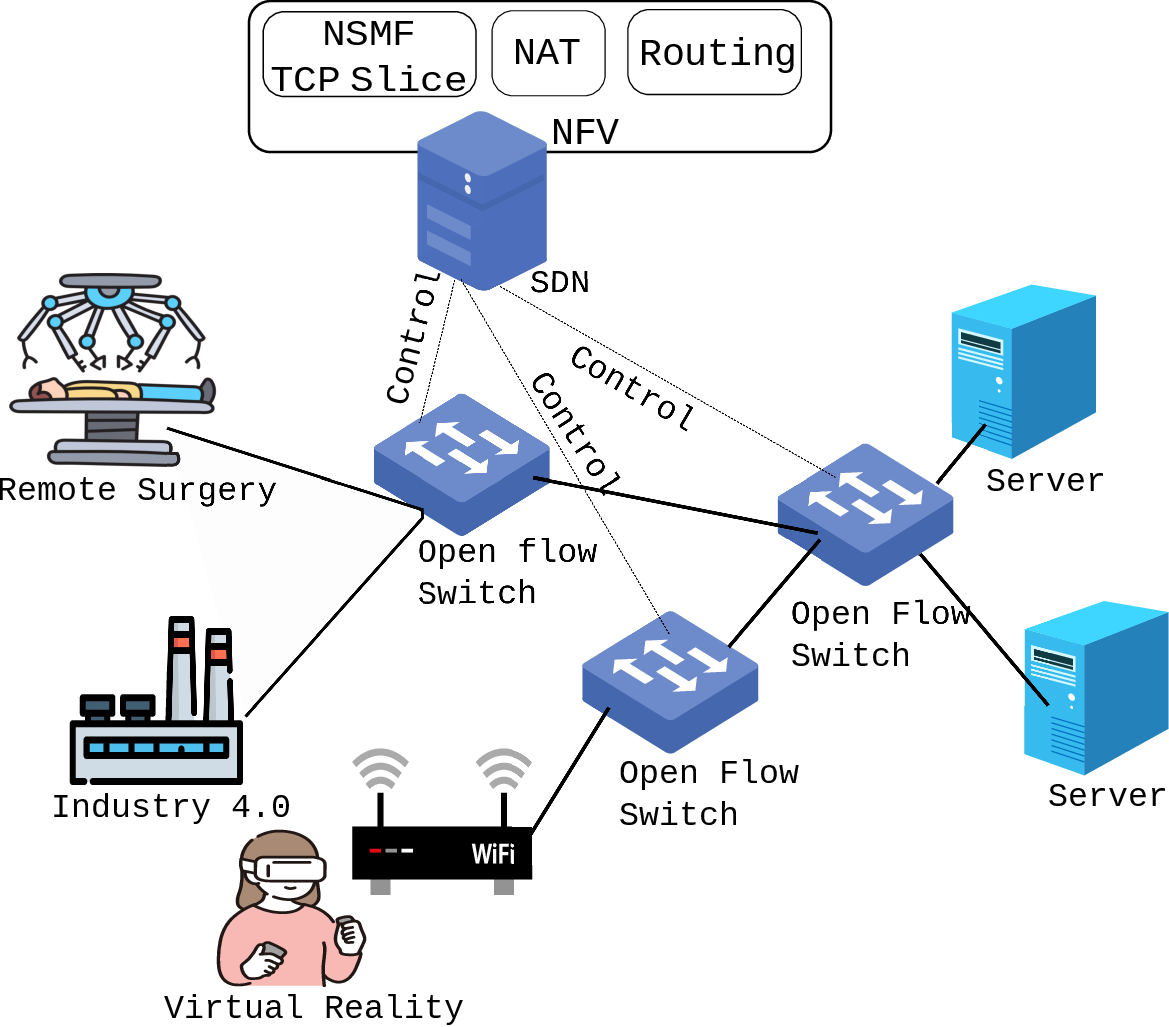}
    \caption{SDN-based network with proposed TCP Slice as a Network Slice Management Function (NSMF) implemented with Network Function Virtualization (NFV).}
    \label{fig:sysmod}
\end{figure}

\subsection{Related Works and Challenges}
TCP, being the fundamental technology for moving internet traffic, is a well-researched topic with vast literature \cite{jacobson1988congestion,jacobson1992tcp,brakmo1994tcp,sally2003highspeed,cardwell2016bbr,cardwell2017bbr}. In this paper, we propose am analytical TCP algorithm that specifically meets delay guarantees which can be employed for both eMBB and uRLLC slice use cases. We resort to the optimization-based approach for congestion control \cite{low1999optimization,kelly1998rate,low2022analytical} that maximizes the total network utility of all users/clients in the network subject to network capacity constraints. The conventional method to yield a distributed control algorithm using such an optimization formulation is to use the dual decomposition technique \cite{luo1993convergence,alghunaim2020linear}. Since the total utility and capacity constraints turn out to be easily separable in the Lagrangian form, the resulting algorithm is simple and exhibits desirable properties. Further, it could be shown that for certain choices of utility functions and parameters, the algorithm follows the widely popular additive increase multiplicative decrease (AIMD) \cite{kelly1998rate} scheme. Variants such as TCP Reno \cite{jacobson1988congestion}, TCP Vegas \cite{brakmo1994tcp} etc., can also be analyzed using the mentioned framework \cite{low2022analytical}.

Although there have been attempts to minimize delay in TCP \cite{luo2017standardization}, to the best of our knowledge, none of them can guarantee delay bounds, specifically due to the following reason: It is evident that the end-to-end delay experienced by a client is influenced not only by its own sending rate but also by the sending rates of other clients sharing network resources. Moreover, the stochastic nature of queuing delay in the network poses challenges in providing strict delay guarantees. Due to these complexities, attempting to incorporate delay constraints directly into the optimization framework using dual decomposition techniques may not be feasible.

To address this issue, we propose to use Network Calculus which offers a suitable toolset to analyze the worst-case bounds on delays and buffer requirements in a network \cite{le2001network}. By incorporating the derived delay bounds as constraints in the optimization problem, we can better manage the allocation of network resources to meet the desired delay requirements for each client. 

% This approach ensures that the network remains stable and efficient while providing improved guarantees for delay and bandwidth satisfaction.

Although Network calculus helps in computing delay bounds, it turns out that such computations require knowledge of global network parameters that are not available at the client site. One may then use an SDN controller to obtain network parameters. However, the controller can only manage the switches/routers in its domain and does not have access to influence sending rates on the client side. Hence, feasible implementation of the derived algorithm  requires the design of a suitable mechanism such that requisite information is disseminated without introducing overheads. We propose to use Explicit Congestion Notification (ECN) which is already a part of the TCP header.

\subsection{Contributions}
Based on the delineated challenges, we present the contributions of the paper.
\begin{itemize}
    \item We develop a model based on network calculus that helps in characterizing the delay bound of a TCP client/source. \textit{To the best of our knowledge, this is the first proposal that incorporates the use of network calculus to compute delay bounds and derives an analytical TCP algorithm using distributed optimization techniques.}
    \item Using the derived bound, we formulate a Network Utility Maximization (NUM) problem with the objective of maximizing the network utility subject to network capacity and delay bound requirements. A dual descent-based solution method is developed which requires the knowledge of global network parameters for implementation. We state convergence results and also show how to account for packetization effects in networks.
    \item We propose a semi-distributed TCP implementation with the help of SDN and OpenFlow switches by marking the ECN bit. We call our algorithm "TCP Slice" since it is dedicated to being used for slices with delay constraints and their management.
    \item We present the steady state and transient performance of TCP Slice and present some important discussions in contrast to fairness-based schemes which are the state of art in the literature. Additionally, we discuss the possibility of a proactive approach without spending any time in the transients that cause significant delays.
\end{itemize}

The rest of the paper is organized as follows: Section 
\ref{sec:modelprob} describes the system model and the problem definition. As described, we append an additional constraint to restrict the delay bound. The relevant background for deriving the proposed bounds is provided in Section \ref{sec:netback}, followed by our proposed bounds in Section \ref{sec:delaycomp}. Thereafter, we reformulate and present the final problem forms in Section \ref{sec:probreform}. Our derived algorithm is proposed and its implementation steps are presented in Section \ref{sec:implementation} followed by relevant results and discussions in Section \ref{sec:results}. Finally, we make some important concluding remarks with future directions in Section \ref{sec:conclusion}.

\section{Model Description and Problem Statement}
\label{sec:modelprob}
We consider a network modeled as a Graph with $G = (V, E)$, where $V$ denotes the set of vertices and $E (\subseteq V^2)$ denotes the set of edges. Some of these vertices are associated with end devices like a mobile, a PC, AR/VR headset/equipment, etc., that run applications. We denote these vertices as the source vertices $V_S \subset V$. For each edge $e$, let $S(e)$ denote the set of sources that use the link $e$. Similarly, let $E(s)$ denote the set of edges or links used by the source $s$. Let $x_s$ denote the transmission rate for source $s$ and $U_s(x_s)$ denotes the utility obtained by the source by transmitting at the rate $x_s$. Let $D_s(x_s)$ denote the delay experienced by the source $s$. Let $x_s \in I_s = [m_s,M_s]$ which denotes the range of feasible rates that the source can transmit. Note that one can treat $m_s$ as a minimum bandwidth constraint on source rate $x_s$. Hence, consideration of minimum bandwidth constraints is implicit in the model. Vector notations are written in bold, sets are denoted in capitals and we use $|A|$ to denote the cardinality or size of the set $A$. The subscript $s$ is used for a source $s$ while $e$ for an edge.
The frequently used notations in the paper are summarized in Table \ref{tab:notations}.

\begin{table}[h]
    \caption{Notations}
    \label{tab:notations}
    \centering
    \begin{tabular}{|M{0.15\columnwidth}|M{0.7\columnwidth}|}
         \hline
         \bfseries Notation & \bfseries Description \\
         \hline
         $s$ & A source/client $s \in V_S$\\
         \hline
         $x_s$ & Rate of source $s$\\
         \hline
         $U_s(x_s)$ & Utility gained by source $s$ by sending at rate $x_s$\\
         \hline
         $n$ & Number of sources $= |V_S|$\\
         \hline 
         $\bm{x}$ & A vector of source rates $(x_1, x_2, \dots, x_{n})$\\
         \hline
         $\bm{x}_{-s}$  & Vector of rates except rate of source $s$ \\
         \hline
         $D_s(\bm{x})$ & Delay experienced by source $s$ due to $x_s$ and $\bm{x}_{-s}$ \\
         \hline
         $d_s$ & Delay bound to be satisfied for source $s$\\
         \hline
         $c_e$ & Capacity of an edge $e$ \\
         \hline
         $E(s)$ & Set of edges used by source $s$ \\
         \hline
         $S(e)$ & Set of sources used by edge $e$ \\
         \hline
         $C_s$ & Set of capacities of edges used by $s$ \\
         \hline
         $c_s^m$ & Minimum capacity in $C_s, ~ (c_s^m = \min C_s)$\\
         \hline
         $\beta_s(t)$ & Effective service curve for source $s$\\
         \hline
         $D_s^{NC}$ & Delay bound for source $s$, calculated using network calculus \\
         \hline
         $x_{-s}$ & Sum of rates of sources other than $s$ \\
         \hline
         $l_{max}$ & Maximum packet size\\
         \hline
         $\sigma$ & Maximum burst size of leaky bucket \\         
         \hline
         $p_e$ & Price computed for a link $e$\\
         \hline
         $p_j$ & Price computed for a source $j$\\
         \hline
         $p_{-s}$ & Sum of prices of sources other than $s$\\
         \hline
         $p^s$ & Sum of prices of links used by source $s$\\
         \hline
         $m_s^e$ & Probability of marking the ECN bit \\
         \hline
    \end{tabular}
    
\end{table}
The following optimization comes into effect where the objective is to maximize the utility subject to delay and physical network constraints:
\subsection{Primal Problem}
\begin{subequations}
\label{eq:primalopt}
\begin{gather}
\max_{x_s \in I_s} \sum_{s} U_s(x_s)\\    
   s.t. ~~ D_s(\bm{x}) = D_s(x_s,\bm{x}_{-s}) \leq d_s ~~\forall s 
    \label{eq:delay}\\
    \sum_{s\in S(e)} x_s \leq c_e ~~\forall e
\end{gather}
\end{subequations}

It is understood that the delay function in \eqref{eq:delay} is dependent on both the source rate and the rate of other sources (that are in the path of the source). This is a complicating non-separable constraint that cannot be easily handled by means of the dual decomposition method, which is commonly used for distributed optimization in case of separable objectives and constraints. To circumvent this issue, we need to find an approximate expression for the delay bound, which would help in obtaining a separable solution. To this purpose, we intend to use network calculus which is a well-known tool for obtaining delay-bound expressions in networks. However, this treatment requires some background, and the relevant results are presented in the section below.

% \subsection{Dual Problem}
% The dual problem can be formulated as:
% \begin{align*}
% L(\bm{x,p}) &= \sum_{s} U_s(x_s) - p_s (D_s(\bm{x}) - d) - \sum_{e} p_e \sum_{s\in S(e)} (x_s - c_e)\\
% &= \sum_s U_s(x_s) - p_s (D_s(\bm{x}) - d) - x_s \sum_{e \in E(s)} p_e + \sum_e p_e c_e\\
% &=\sum_s U_s(x_s) - p_s (D_s(\bm{x})-d) - x_s \sum_{e \in E(s)} p_e + \sum_e p_e c_e\\
% &=\sum_s U_s(x_s) - p_s (D_s(\bm{x})-d) - x_s p^s + \sum_e p_e c_e 
% \end{align*}

% $E(s)$ denotes the set of edges along the paths used by source $s$. $\bm{x} = (x_1,x_2,\dots,x_{|S|})$ and $\bm{p} = (p_s,p_e) ~\forall s,e$.
% \begin{equation}
%   p^s = \sum_{e \in E(s)} p_e  
% \end{equation}

% The objective function of dual problem is:
% \begin{equation}
%   Y(\bm{p}) = \max_{\bm{x}} L(\bm{x,p}) = \max_{\bm{x}} \sum_s U_s(x_s)-x_s p^s - \sum_s{p_s(D_s(\bm{x}) -d)}  
% \end{equation}

% Consider two sources $x_1$ and $x_2$
% \begin{equation}
%     \max_{x_1,x_2} U_1(x_1)-x_1 p^1 + U_2(x_2)-x_2 p^2 - {p_1(D_1(x_1,x_2) -d)+p_2(D_2(x_1,x_2) -d)}  
% \end{equation}

% The function $\sum_s U_s(x_s)-x_s p^s$ is separable but $p_s(D_s(\bm{x}) -d)$ are not. 
% Consider two sources $x_1$ and $x_2$.
% We want to find 
% $$\min_{x_1,x_2} p_1(D_1(x_1,x_2) -d) + p_2(D_2(x_1,x_2) -d)$$

% Then, to do this we introduce independent variables $y_1,y2$ such that
% \begin{align}tocken-bucket based
% &\min_{x_1,x_2, y_1,y_2}p_1(D_1(x_1,x_2) -d) + p_2(D_2(y_1,y_2) -d)\\
% s.t.& \nonumber\\
% &y_1 = x_1\\
% &y_2 = x_2
% \end{align}

% \begin{equation}
%   Y(\bm{p}) = \max_{\bm{x}} L(\bm{x,p}) = \sum_s B_s(p_s,p^s)+\sum_e p_e c_e  
% \end{equation}

% where

\section{Network Calculus}
 We will present the relevant results in the theory which helps us to obtain an expression for the delay bound in a network.
\subsection{Background}
\label{sec:netback}
Evidently, ensuring delay and backlog bounds require restrictions on the arrival and service processes. Consider the description of the generated traffic at a node by the cumulative process $A(t)$ which provides the total number of bits/bytes/packets arriving at the node until time $t$. Typically, the arrivals are restricted using arrival curves with the following definition. 

\begin{definition}[Arrival Curve \cite{le2001network,chang2000performance}]
    An cumulative arrival process $A(t)$ is said to have an arrival curve $\alpha(t)$ iff $A(t) - A(s) \leq \alpha(t-s), ~ \forall ~ 0\leq s\leq t \Leftrightarrow A(t) \leq \inf_{s\in [0,t]} [A(s) + \alpha (t-s)]$
\end{definition}

From the definition of an arrival curve, it is clear that an arrival curve restricts the number of packets generated in a given interval. A typical and practical example of an arrival curve is the output of a leaky-bucket implementation at the source. If a source implements a leaky bucket to the arriving traffic, that operates at the rate of $x$ and bucket length (or buffer length) $\sigma$, the arrival curve is given by :

$$\alpha(t) = \sigma + x t$$

A leaky or token bucket operates in the following way: The arriving packets are stored in a buffer (theoretically of infinite size). The bucket is a separate buffer with capacity $\sigma$, in which tokens are generated at rate $x$. Each packet that arrives at a leaky bucket, finds a token in the token bucket, takes the token, and leaves the bucket immediately. Since the bucket size is limited by $\sigma$ and the number of tokens that are generated in an interval of $(t-s)$ is $x(t-s)$, the maximum number of packets that can go out of this system in an interval of $(t-s)$ is $\sigma + x(t-s)$. Hence, $\alpha(t) = \sigma + x t$. Often this arrival curve is termed to be $(\sigma,x)$-upper constrained.
\begin{lemma}[Aggregate Multiplexing \cite{le2001network,chang2000performance}]
    If two flows $A_1$ and $A_2$ are aggregated at a node, and $\alpha_1$, $\alpha_2$ are their arrival curves, then $\alpha_1 + \alpha_2$ is an arrival curve of the aggregated flow $A_1 + A_2$.
\end{lemma}

Similar to restricting the arrivals, the service process can also be defined in terms of service curves. If $t_0$ denotes the last time instant from $t$ when the server becomes busy (or backlog starts), we must have $D(t) - D(t_0) = c(t-t_0)$. By definition of backlog, it must be that $D(t_0) = A(t_0)$ and hence $D(t) - A(t_0) = c(t- t_0) \Rightarrow D(t) = A(t_0) + c(t- t_0)\geq \inf_{s\in [0,t]} [A(s) + c(t-s)] $.

\begin{definition}[Service Curve \cite{le2001network}]
    If $B(t)$ be the cumulative departure and $A(t)$ the cumulative arrival, then $\beta(t)$ is said to be a service curve iff
    $D(t) \geq \inf_{s\in[0,t]} A(s) + \beta(t-s), ~ \forall ~ t$
\end{definition}
\begin{definition}[Strict Service Curve \cite{le2001network}]
We say that system S offers a strict service curve $\beta$ to a flow if,
during any backlogged period of duration $u$, the output of the flow is at least equal to $\beta(u)$. Every strict service curve is a service curve.
\end{definition}
The service curve ensures that a minimum number of bits/bytes/packets are served in an interval. For a work conserving (server is in operation whenever the buffer is non-empty) constant rate server or a link with serving rate $c$, the strict service curve is 
$\beta(t) = ct$. In network calculus, often we represent these types of service curves with a generalized form of

$$\beta_{R,T}(t) = R[t-T]^+$$

where $[x]^+ = \max\{0,x\}$. These are also termed as rate-latency servers.

\begin{definition}[Min-Plus Convolution ($\otimes$) \cite{le2001network}]
We observe that both arrival and service curves involve calculating $\inf_{s\in [0,t]} [A(s) + h(t-s)]$. Analogous to convolution in filtering theory, this operation can also be perceived as convolution, where the integral or sum is replaced by an infimum (minimum), and the product is replaced by addition or plus. Thus, this operation is termed as min-plus convolution denoted by the operator $\otimes$. 
$$A \otimes h = \inf_{s\in [0,t]} [A(s) + h(t-s)]$$
\end{definition}
We present a result termed as concatenation theorem which simply states that the concatenation of network elements (like servers) leads to the convolution of their service curves similar to cascading of filters in signal processing.

\begin{theorem}[Concatenation of Nodes \cite{le2001network}]
Assume a flow traverses systems $S_1$ and $S_2$ in sequence. Assume that $S_1$ offers a service curve of $\beta_i$, $i=1,2$ to the flow. Then the concatenation of the two systems offers a service curve of $\beta_1 \otimes \beta_2$ to the flow.
 
\end{theorem}

In addition to the mentioned theorem, we also present a result related to finding the convolution of piecewise linear convex functions.
\begin{lemma}[Convolution for Piecewise Linear Convex Functions \cite{le2001network}]
If $f$ and $g$ are convex and piecewise linear $f \otimes g$ is obtained by putting end-to-end the different linear pieces of $f$ and $g$, sorted by increasing slopes.
\end{lemma}

\begin{corollary}[Concatenation of Rate-latency servers \cite{le2001network}]
If $\beta_1 (t) = \beta_{R_1,T_1} (t) = R_1[t - T_1]^+$ and  $\beta_2 (t) = \beta_{R_2,T_2} (t) = R_2[t - T_2]^+$, then $\beta_1(t) \otimes \beta_2(t)$ is given by:
\begin{align*}
\begin{split}
\beta_1 \otimes \beta_2 &= \min\{R_1, R_2\} \left[t - (T_1 + T_2)\right]^+\\
&= \beta_{\min\{R_1,R_2\},(T_1 + T_2)}(t)
\end{split}    
\end{align*}
\end{corollary}

The proof is a straightforward application of Lemma 2.
The result can be easily generalized to concatenate $n$ such rate-latency servers. 
Our objective is to find a worst-case delay bound for a source in a network. 

\textbf{Assumption:} \textit{Each source implements a token/leaky bucket to ensure an arrival curve.}\footnote{This is a fair assumption since TCP flow control is implemented with the help of leaky buckets and traffic shaping is an integral part of rate control. Further, while admitting slices, usually there are service level agreements that have information regarding traffic characteristics. Once the characteristic is known, it is customary to design a flow controller that allows the smooth operation of networks.} 

This assumption allows us to use deterministic network calculus to compute delay bounds using the following theorems \cite{le2001network}.
\begin{theorem}[Blind Multiplexing \cite{le2001network}]
 Consider a node serving two flows, 1 and 2, with some unknown arbitration (scheduling policy) between the two flows. Assume that the node guarantees a strict service curve $\beta$ to the aggregate of the two flows. Assume that $\alpha_2$ is an arrival curve for flow 2. Define $\beta_1(t) := [\beta(t) − \alpha_2 (t)] ^+$ . If $\beta_1$ is wide-sense
increasing, then it is a service curve for flow 1.
\end{theorem}

If $\beta(t) = ct$, and $\alpha_2(t) = \sigma_2 + x_2 t$, we have, $\beta_1(t) = [ct - \sigma_2 - x_2 t]^+ = (c-x_2)\left[t - \frac{\sigma_2}{c-x_2}\right]^+ $ (assuming $c>x_2$ for stability). Service curves of the form $\beta(t) = R[t-T]^+$ are known as rate-latency curves with rate $R$ and latency $T$.

Typically, wired communication links are work-conserving servers implying that the links are operational at full capacity whenever the buffer is non-empty, or else the server is non-operational. The following theorem states that the output of work-conserving servers are $(\sigma,x)$ constrained if the input is so.

\begin{lemma}[\cite{chang2000performance}]
    Let $A$ and $B$ be the input and the output of a network element. Suppose that $A$ is $(\sigma, x)$ - upper constrained. If the network element is a work-conserving link, then B is also $(\sigma, x)$-upper constrained.
\end{lemma}

We are now ready with the requisite theory to derive the delay bound in a network and use it to perform the optimization \eqref{eq:primalopt}.
% The mentioned theorems for network calculus hold true for a fluid-based model. However, 
\subsection{Computing Delay Bounds}
\label{sec:delaycomp}
Given the graph model under consideration, we intend to obtain a delay constraint of a source. With the given graph model, the routing information is deemed to be available and hence we can assume that at any given time, the path used from source to destination is known. We consider that a source $s \in S$ obeys an arrival curve $(\sigma_s,x_s)$. Let $C_s$ be the set of link capacities used by the source $s$ to reach its destination and hence the link capacities may be enumerated as $C_s = \{c_s^1, \dots, c_s^{|C_s|}\}$. There may be flows arriving and departing from each of these links. There may be two ways to compute the overall delay of a source through a path. The first way is to over-estimate the worst-case, requiring no knowledge about the exact sources whose flows are going through a link. We assume that all links in the path of source $s$ are subject to all flows. The overall arrival curve of the remaining is $(\sum_{j\neq s, j\in S} \sigma_j, \sum_{j\neq s, j\in S} x_j)$-upper-constrained. We may denote this as $(\sigma_{-s}, x_{-s})$. By, Theorem 2, we may compute the effective service at a link with capacity $c_s^k$ to be $\beta_s^k(t) = (c_s^k - x_{-s})\left[t - \frac{\sigma_{-s}}{c_s^k - x_{-s}}\right] $. By Theorem 1, we may find the effective service curve as 
$$\beta_s(t) = \beta_s^1(t)\otimes \dots \otimes \beta_s^{|C_s|}(t)$$

Considering that all sources are identical w.r.t the amount of burst it can handle i.e. $\sigma_s = \sigma ~ \forall s $, we have the following proposition
\begin{proposition}
For sources with identical token bucket size $(\sigma_s = \sigma ~ \forall s )$, the effective service rate for a flow from source $s$ is given by

\begin{equation}
    \beta_s(t) =  (c_s^m - x_{-s})\left[t - \sum_{k=1}^{|C_s|}\frac{(n-1)\sigma}{c_s^k-x_{-s}}\right]^+ , ~ c_s^m = \min C_s
    \label{eq:effectserv}
\end{equation}
\end{proposition}
\begin{proof}
The proof is an application of Corollary 2.1.
\end{proof}

\begin{proposition}
The worst-case delay for source $s$, $D_s$ can be obtained as 
\begin{equation}
    D_s^{NC} = \frac{\left(|C_s|(n-1)+1\right)\sigma}{(c_s^m - x_{-s})}, \quad c_s^m = \min C_s
    \label{eq:delaybound}
\end{equation}
\end{proposition}
\begin{proof}
    To find the delay bound, we would compute the largest possible horizontal distance between the arrival curve and the service curve.
\begin{align}
\sigma &=  (c_s^m - x_{-s})\left[D_s^{NC} - \sum_{k=1}^{|C_s|}\frac{(n-1)\sigma}{c_s^k-x_{-s}}\right]^+\\
\begin{split}
D_s^{NC} &=  \frac{\sigma}{(c_s^m - x_{-s})} + \sum_{k=1}^{|C_s|}\frac{(n-1)\sigma}{c_s^k-x_{-s}} 
\\
&\leq \frac{\sigma}{(c_s^m - x_{-s})} + \sum_{k=1}^{|C_s|}\frac{(n-1)\sigma}{c_s^m-x_{-s}} ~ (\because c_s^k \geq c_s^m, \forall k)\label{eq:ineq}\\ 
&=  \frac{\left(|C_s|(n-1)+1\right)\sigma}{(c_s^m - x_{-s})}
\end{split}
\end{align}
\end{proof}
An alternate approach would be to have the information regarding exact flows through the links and evaluate the effective service curve. An SDN controller can calculate the exact effective service curve and compute the tightest possible delay bound. However, to keep the model simple and intuitive, a better bound can be obtained by using \eqref{eq:delaybound} itself. 
\begin{proposition}
A tighter delay bound for source $s$, $(D_s^{NC})$ can be obtained as 
\begin{equation}
    D_s^{NC} = \frac{\left(|C_s|\left(\max_{k\in E(s)}|S(k)|-1\right)+1\right)\sigma}{\min_{k\in E(s)}(c_s^k - \sum_{j\neq s, j\in S(k)} x_j)}
    \label{eq:revdelaybound}
\end{equation}
\end{proposition}
\begin{proof}
    We look to replace $n$ in the numerator by the maximum number of flows that $s$ interacts with on a link and the denominator by the minimum possible residual capacity. Recall that $E(s)$ denotes the set of edges used by source $s$ and $S(e)$ denotes the set of sources that use the link $e$. In the numerator we replace $n$ (in \eqref{eq:delaybound}), the total number of sources by the maximum number of sources that use a link $e \in E(s)$. This is the maximum possible burst that the arrival curve of $s$ would face in its path. In the denominator, we replace the overestimated $(c_s^m - x_{-s})$ by the residual capacity of the most congested link among the links used by $s$. At any given moment, this is the worst possible left-behind capacity for flow $x_s$ (the overestimate is even worse and might not even occur in practice as we consider all flows are mixed with others). Hence, the resulting expression will always be an upper bound on the delay. However, this may not be the suprema and hence is not the least upper bound.
\end{proof}

\subsection{Effect of Packetization}

The Network calculus results used for obtaining the proposed distributed algorithm hold true for a fluid-based model. However, in a network, one is typically interested in per-packet delays. Packetization has an important implication in terms of network calculus as the results can no longer be directly applied and a packetizer \cite{le2001network} must be introduced at every server (or link in this case). Considering a packetizer changes the service curves for each link. As shown in \cite{le2001network}, the service curve for a server operating at a rate $R$ followed by a packetizer is $
\beta(t) = R[t - \frac{l_{max}}{R}]^+$, where the maximum length of a packet $l_{max}$. Thus, in this case,  $
\beta(t) = c[t - \frac{l_{max}}{c}]^+$, $
\alpha_2 (t) = \sigma_2 + x_2 t$, we have $\beta_1 (t) = \left[ct - \sigma_2 - x_2 t - \l_{max}\right]^+ = (c-x_2)\left[t - \frac{\sigma_2 + l_{max}}{c - x_2}\right]^+$. Thus, the effective service curve \eqref{eq:effectserv} changes to 

\begin{equation}
    \beta_s(t) = (c_s^m - x_{-s}) \left[t - \sum_{k=1}^{|C_s|}\frac{(n-1)\sigma+l_{max}}{c_s^m - x_{-s}}\right]^+  
\end{equation}
where $c_s^m = \min C_s$. Thus, \eqref{eq:revdelaybound} should be modified as 
\begin{equation}
    D_s^{NC} = \frac{\left(|C_s|\left(\max_{k\in E(s)}|S(k)|-1\right)+1\right)\sigma + |C_s|l_{max}}{\min_{k\in E(s)}(c_s^k - \sum_{j\neq s, j\in S(k)} x_j)}
    \label{eq:revpackdelaybound}
\end{equation}

% Hence, the obtained delay-bound changes to 
% $$D_s = \frac{n \sigma + (n-1) l_{max}}{c_s^k - x_{-s}}$$

% Equivalently, the constraint changes to 

% $$\sum_{j\neq s} x_j \leq \min C_s - \frac{n\sigma + (n-1) l_{max}}{d_s}$$

For simplicity, we proceed with \eqref{eq:delaybound}, however the presented ideas can be extended for \eqref{eq:revpackdelaybound} as well and is used for our simulations.
\section{Network Utility Maximization with Delay Constraints}
\label{sec:probreform}
Based on the developed model, we now concentrate on the problem of our concern. Using \eqref{eq:delaybound}, we may obtain a constraint for each source as follows:
Let $d_s$ denote a delay constraint that needs to be satisfied.

$\frac{\left(|C_s|(n-1)+1\right)\sigma}{(c_s^k - x_{-s})} \leq d_s \Rightarrow x_{-s} \leq \min C_s - \frac{\left(|C_s|(n-1)+1\right)\sigma}{d_s}$
% \frac{n\sigma}{\min C_s - x_{-s}} \leq d_s \Rightarrow \sum_{j\neq s} x_j  \leq \min C_s - \frac{n\sigma}{d_s} \quad \forall s$$. 
Hence, the primal problem changes to
\subsection{Primal Problem}
\begin{subequations}
\begin{gather}
\max_{x_s \in I_s} \sum_{s} U_s(x_s)\\
   s.t. ~~ \sum_{j\neq s} x_j  \leq \min C_s - \frac{\left(|C_s|(n-1)+1\right)\sigma}{d_s}  ~~\forall s 
   \label{eq:moddelay}
   \\
    \sum_{s\in S(e)} x_s \leq c_e ~~\forall e
\end{gather}
\label{eq:modopt}
\end{subequations}
As we can observe now, by using Network calculus-based delay bound, we are able to write \eqref{eq:delay} in terms of the sum of source rates \eqref{eq:moddelay}. We are now in a position to use dual decomposition technique in order to obtain a distributed solution.
\subsection{Dual Problem}
The dual problem can be formulated as:
\begin{align*}
\begin{split}
L(\bm{x,p}) 
% &= \sum_{s} U_s(x_s) - \sum_s p_s \Bigg(\sum_{j\neq s} x_j  - \min C_s \\
% & \quad + \frac{\left(|C_s|(n-1)+1\right)\sigma}{d_s}\Bigg) - \sum_{e} p_e \sum_{s\in S(e)} (x_s - c_e)
% \end{split}
% \\
% \begin{split}
% &=\sum_{s} U_s(x_s) - \sum_s \Bigg(x_s\sum_{j\neq s} p_j  - p_s\min C_s \\ & +p_s\frac{\left(|C_s|(n-1)+1\right)\sigma}{d_s}\Bigg) - \sum_s x_s \sum_{e \in E(s)} p_e + \sum_e p_e c_e
% \end{split}
% \\
% &=\sum_{s} \left[U_s(x_s) - x_s\sum_{j\neq s} p_j  + p_s\min C_s -p_s\frac{n\sigma}{d_s}\right] - \sum_{e} p_e \sum_{s\in S(e)} (x_s - c_e)\\
% &=\sum_{s} \left[U_s(x_s) - x_s\sum_{j\neq s} p_j  + p_s\min C_s -p_s\frac{n\sigma}{d_s}\right] - \sum_s x_s \sum_{e \in E(s)} p_e + \sum_e p_e c_e\\
% \begin{split}
&=\sum_{s} \Bigg[U_s(x_s) - x_s\sum_{j\neq s} p_j 
+ p_s\min C_s - \\ 
& \quad p_s\frac{\left(|C_s|(n-1)+1\right)\sigma}{d_s} - x_s \sum_{e \in E(s)} p_e\Bigg] + \sum_e p_e c_e
\end{split}
% &= \sum_s \left[U_s(x_s) -  p_s \left(\sum_{j\neq s} x_j  - \min C_s + \frac{n\sigma}{d_s}\right) - x_s \sum_{e \in E(s)} p_e\right] + \sum_e p_e c_e\\
% &=\sum_s \left[U_s(x_s) -  p_s \left(\sum_{j\neq s} x_j  - \min C_s + \frac{n\sigma}{d_s}\right) - x_s \sum_{e \in E(s)} p_e \right]+ \sum_e p_e c_e\\
% &=\sum_s \left[U_s(x_s) -  p_s \left(\sum_{j\neq s} x_j  - \min C_s + \frac{n\sigma}{d_s}\right) - x_s p^s\right] + \sum_e p_e c_e 
\end{align*}

% where $p^s = \sum_{e \in E(s)} p_e  $. $E(s)$ denotes the set of edges along the paths used by source $s$. $\bm{x} = (x_1,x_2,\dots,x_{|S|})$ and $\bm{p} = (p_s,p_e) ~\forall s,e$.

%
% This equation is not separable in general. Since we want to use provide a bound on delay, the bound $D_s(x_s,\bm{x_s})\leq d$ is automatically satisfied if we use the following bound instead:
% \begin{equation}
%     D_s(x_s,\bm{x_{-s}})\leq D^a_s(x_s,\max_{e\in E(s)} c_e - x_s) \leq d
% \end{equation}
% $D^a_s$ is a delay function by which we summarize the effect of delays of all other sources by a single variable $\max_{e\in E(s)} c_e - x_s$. We claim that such a function exists and would help us in proposing a control algorithm.
% To show this we take the help of network calculus.
% \begin{proposition}
% There exists a monotonically increasing function in $D_s^a(x_s)$ such that $D_s^a(x_s) \geq D_s(x_s, \bm{x_{-s}})$
% \end{proposition}
% \begin{proof}
% We use network calculus to prove this for a wide class of scheduling and networks.
% \end{proof}

where $p_s$ and $p_e$ correspond to prices or Lagrange multipliers for each source and link. We may define,
\begin{equation}
\begin{split}
  B_s(p_s,p_e) &= \max_{x_s\in I_s} U_s(x_s) -  x_s \left(\sum_{j\neq s} p_j +  \sum_{e\in E(s)}p_e \right)\\ 
  &=\max_{x_s\in I_s} U_s(x_s) -  x_s (p_{-s} + p^s) 
  \label{eq:sourceopt}
  \end{split}
\end{equation}
where $p_{-s} = \sum_{j\neq s} p_j$ and $p^s = \sum_{e\in E(s)}p_e$. $B_s(p_s,p_e)$ can be interpreted as the effective utility of a source which is the difference of its utility and prices associated with violation of delay and capacity constraints. The optimal rate would try to maximize this difference.  
Hence, 
\begin{equation*}
\begin{split}
\max_{\bm{x}} L(\bm{x},\bm{p}) &= Y(\bm{p}) =  \sum_s B_s(p_s,p^s)\\ &+ 
\sum_s p_s \left(\min C_s - \frac{\left(|C_s|(n-1)+1\right)\sigma}{d_s}  \right) + \sum_e p_e c_e
\end{split}
\end{equation*}
Thus, the dual problem is:

\begin{equation}
\min_{p_e, p_s \geq 0, ~\forall e,s}Y(\bm{p})
\label{eq:dual}
\end{equation}
 In this model, each source pays two kinds of prices $p_{-s}$ and $p^s$. While $p_{-s}$ is associated with the price paid in regards to delay, $p^s$ is the price paid for using each link per unit bandwidth. 

% \textbf{Assumption 1:} \textit{The function $D_s(\bm{x})$ is monotonically increasing and convex in $x_s$.}

\textbf{Assumption:} \textit{The function $-U_s(x_s)$ is continuous, differentiable and strongly convex in $x_s$.}

This assumption allows us to solve \eqref{eq:sourceopt} uniquely. We have:

\begin{equation}
    x_s^*(\bm{p}) = U_s'^{-1}(p_{-s}+p^s)
    \label{eq:sourcerate}
\end{equation}

% \textbf{Assumption 2:} \textit{The function $D_s(\bm{x})$ is monotonically increasing and convex in $x_s$.}

% Given these assumptions, for given $(p_s,p^s,\bm{x}_{-s})$  there is always a unique maximizer $x_s(p_s,p^s,\bm{x}_{-s})$ for \eqref{eq:sourceopt}. Here $\bm{x}_{-s}$ denotes the rates decided by all other sources.

% However, this is not the solution to \eqref{eq:flowopt}. We need to solve the dual problem \eqref{eq:dual} to obtain $(p_s^*,p^{*s})$. For this $(p_s^*,p^{*s})$, we need to find the optimal $x_s^*(p_s^*,p^{*s})$, which would be the solution to \eqref{eq:flowopt}. 

% For a given $\bm{p}$, the critical point $x_s(\bm{p})$ is given as
% \begin{align*}
%     &U'_s(x_s) - p_s D'_s(x_s) -p^s = 0\\
%     \Rightarrow & x_s = 
% \end{align*}

\section{Semi-distributed TCP Slice Algorithm}
\label{sec:algo}
We solve the dual problem \eqref{eq:dual} using gradient projection method for source prices and link prices respectively. 

\subsection{Solution using Dual Descent}
For a given price vector $\bm{p}$, each source obtains $x_s^*(\bm{p})$.

\begin{align}
    p_e(t+1) &= \left[p_e(t) - \gamma\frac{\partial Y(\bm{p}(t))}{\partial p_e}\right]^+ ~\forall e\in E\\
    p_s(t+1) &= \left[p_s(t) - \gamma\frac{\partial Y(\bm{p}(t))}{\partial p_s}\right]^+ ~~~\forall s\in S
\end{align}

We have:
\begin{align*}
% \begin{split}
 Y(\bm{p}) &=  
 % \sum_s B_s(p_s,p^s) + \\ 
 % & \quad \sum_s p_s \left(\min C_s - \frac{\left(|C_s|(n-1)+1\right)\sigma}{d_s}\right) + \sum_e p_e c_e
 % \end{split}
 % \\
 % \begin{split}
 % &=\sum_s \max_{x_s\in I_s} U_s(x_s) -  x_s (p_{-s} + p^s) +  \\
 % & \quad \sum_s p_s \left(\min C_s - \frac{\left(|C_s|(n-1)+1\right)\sigma}{d_s}\right) + \sum_e p_e c_e
 % \end{split}
 % \\
 % \begin{split}
 % &=
 \sum_s U_s(x_s^*) -  x_s^* (p_{-s} + p^s) +\\
 & \quad \sum_s p_s \left(\min C_s - \frac{\left(|C_s|(n-1)+1\right)\sigma}{d_s}\right) + \sum_e p_e c_e
 % \end{split}
\end{align*}

% \begin{align*}
%   Y(\bm{p}) = \max_{x_s \in I_s} L(\bm{x,p}) &= \sum_s B_s(p_s,p^s)+\sum_e p_e c_e\\  
%   &= \sum_s \max_{x_s\in I_s} U_s(x_s) - p_s (D_s(x_s)-d) - x_s p^s+\sum_e p_e c_e\\
%   &= \sum_s U_s(x_s(p_s,p^s)) - p_s (D_s(x_s(p_s,p^s))-d) - x_s(p_s,p^s) p^s+\sum_e p_e c_e\\
% \end{align*}

where $x_s^*$ is the maximizing $x_s$. This must satisfy:

\begin{equation}
    \frac{\partial U_s(x_s^*)}{\partial x_s^*}  - (p_{-s} + p^s) = 0
\end{equation}

We need to differentiate $Y$ w.r.t the price variables for each source and link. 
% Differentiating w.r.t $p_s$

% \begin{align*}
%     \frac{\partial Y(\bm{p})}{\partial p_s} &= \sum_s \left(\frac{\partial U_s(x_s^*)}{\partial p_s}  - \frac{\partial x_s^*}{\partial p_s} (p_{-s} + p^s)\right) - \left(\min C_s - \frac{n\sigma}{d_s}\right)\\
%     &= \sum_s\left(\frac{\partial U_s(x_s^*)}{\partial x_s^*}\frac{\partial x_s^*}{\partial p_s}  - \frac{\partial x_s^*}{\partial p_s} (p_{-s} + p^s)\right)  - \left(\min C_s - \frac{n\sigma}{d_s}\right)\\
%     &=-\left(\min C_s - \frac{n\sigma}{d_s}\right)
% \end{align*}
Differentiating $Y$ w.r.t to other source variable say $p_j, j\neq s$
\begin{align*}
% \begin{split}
    \frac{\partial Y(\bm{p})}{\partial p_j} 
    % &= \sum_s\left(\frac{\partial U_s(x_s^*)}{\partial x_s^*} 
    % \frac{\partial x_s^*}{\partial p_j}- \frac{\partial x_s^*}{\partial p_j} (p_{-s} + p^s) - x_s^*\right)  + \\
    % & \quad \left(\min C_j - \frac{\left(|C_j|(n-1)+1\right)\sigma}{d_j}\right)
    % \end{split}
    % \\
    &= -\sum_{s \neq j} x_s^* + \left(\min C_j - \frac{\left(|C_j|(n-1)+1\right)\sigma}{d_j}\right)
\end{align*}

Differentiating $Y$ w.r.t to link price $p_e$
\begin{align*}
    \frac{\partial Y(\bm{p})}{\partial p_e} 
    % &= \sum_s \left(\frac{\partial U_s(x_s^*)}{\partial p_e}  - \frac{\partial x_s^*}{\partial p_e} (p_{-s} + p^s) - x_{s\in S(e)}^* \right) +c_e\\
    % &= \sum_s \left(\frac{\partial U_s(x_s^*)}{\partial x_s^*}\frac{\partial x_s^*}{\partial p_s}  - \frac{\partial x_s^*}{\partial p_s} (p_{-s} + p^s) - x_{s\in S(e)}^*\right)\\
    % & \quad +c_e\\
    &=c_e - \sum_{s\in S(e)} x_s^*
\end{align*}
Thus, we have:

\begin{align}
    \begin{split}
        p_{j\neq s}(t+1) 
        & = \Bigg[p_j(t) - \gamma_j  \bigg(\min C_j - \\
        & \quad \frac{\left(|C_j|(n-1)+1\right)\sigma}{d_j} - \sum_{s \neq j}x_s^* \bigg)  \Bigg]^+ \label{eq:delayupdate}
    \end{split}
    \\
    p_e(t+1) &= \left[p_e(t) - \gamma_e \left(c_e - \sum_{s\in S(e)} x_s^*\right)  \right]^+ \label{eq:linkupdate}
\end{align}

Each source computes $x_s(\bm{p})$ for a given $\bm{p}$ using \eqref{eq:sourcerate}. Links $e$ will update their prices $p_e$ following \eqref{eq:linkupdate}. Then, sources should update the prices $p_j$ using \eqref{eq:delayupdate}. Note that in \eqref{eq:delayupdate}, each time the price is decremented by $\gamma_j$ times the difference of the denominator and numerator of \eqref{eq:delaybound}. The same principle is followed while implementing \eqref{eq:revpackdelaybound}.

% However, a source $s$ needs to compute only $p_{-s}$ and $p^s$, Let us evaluate these sums:

% \begin{align}
%     p_{-s} (t+1)= \sum_{j\neq s} p_j(t+1) &= \sum_{j \neq s}\left[p_j(t) - \gamma_j  \left(\min C_j - \frac{n\sigma}{d_j} - \sum_{s \neq j}x_s^*(t) \right)  \right]\\
%     &=p_{-s}(t) - \sum_{j \neq s}\gamma_j  \left(\min C_j - \frac{n\sigma}{d_j} - \sum_{s \neq j}x_s^*(t) \right)
%     \label{eq:sourceupdate}
% \end{align}

From \eqref{eq:delayupdate}, it is evident that each source requires the knowledge of the minimum capacity edges of other sources and also the sum rates of other users, to compute its own rate. However, this information is not readily available to a source node. Thus, a mechanism has to be developed that facilitates a source node to update its prices. To facilitate this, a random exponential marking (REM) scheme is designed similar to the one described in \cite{athuraliya2001rem}.  Although the queues/buffers on the edges can compute $p_e(t)$ independently and implement \eqref{eq:linkupdate} as described in  \cite{athuraliya2001rem}, now the same must be aided with additional adjustments so that \eqref{eq:delayupdate} is not neglected. We propose to perform this with help of SDN controller, which controls the intermediate openflow switches.

\subsection{Implementation with SDN controller}
\label{sec:implementation}
We now devise a scheme that may be implemented with help of an SDN controller that has global information regarding the entire network and routing states. Specifically, we propose SDN to compute \eqref{eq:delayupdate} and use the computed prices to update the marking probability at the OpenFlow switches.

For a given source $s$, it needs to update its rate according to \eqref{eq:sourcerate} which requires computing $p_{-s} + p^s$. To do this, we impose an exponential marking probability distribution for each edge by which a packet from source $s$ traverses. We define $m_s^e(t)$ to be the probability of marking a packet at edge $e$. Then, the probability of a packet being marked after traveling through a path (say for source $s$) can be computed by considering the probability that the packet is marked in any one of the links used by source $s$. The marking information is fed back to the source through the TCP acknowledgment facility which should allow the source to compute \eqref{eq:sourcerate}. Hence we want the end-to-end marking probability to capture the sum $p_{-s} + p^s$. This leads to the following equations

\begin{gather}
     1-\prod_{e\in E(s)}^{} \left(1-m_s^e(t)\right) = 1- e^{-\left(p_{-s}+ p^s\right)} \nonumber \\
    \Rightarrow m_s^e(t) = 1 - e^{-\left(\frac{\sum \limits_{j \neq s} p_j}{|E(s)|} + p_e \right)}
    \label{eq:markprob}
\end{gather}

%% This declares a command \Comment
%% The argument will be surrounded by /* ... */

We present our semi-distributed algorithm. We break our algorithm into two parts. A source starts sending packets in the network with its minimum capacity $m_s$. Once the packets reach SDN switches, they communicate with the SDN regarding flow setup. Since we are concerned with slices, the a network function named network slice management function is run at the SDN to monitor the slices. We propose that TCP slice can be a part of this function. Over time, the switches collect statistics regarding the data rate of users and send them to the SDN for management purposes. The SDN then computes the prices $p_j$ for each source using \eqref{eq:delayupdate} and sends it to the switches for marking the ECN bit. The price $p_e$ \eqref{eq:linkupdate} can either be computed at the switches or; the SDN, being a central node, can also compute the same. The probability of marking a packet from source $s$ is computed using \eqref{eq:markprob} and switches use this probability to set the ECN bit. Algorithm \ref{algo:switch} describes these operations. On receiving the packets, the source can keep computing the probability of marked packets and compute the price by an inverse transformation. Once the prices are computed by the source, it uses \eqref{eq:sourcerate} to find its rate. These steps are demonstrated by Algorithm \ref{algo:source}. 
\SetKwComment{Comment}{/* }{ */}

\begin{algorithm}
\caption{Algorithm to calculate rates for a source $s$}\label{alg:two}
\SetKwInOut{KwIn}{Input}
\SetKwInOut{KwOut}{Output}
\KwIn{$x_s(t)$}
\KwOut{$x_s(t+1)$}
Compute Probability of marked packets using received ECN bit 

$P_M(t) = \frac{\text{\# marked packets received till t+1}}{\text{\# packets received till t+1}}$\;
Compute the sum of link prices $p_{-s}(t) + p^s(t) = -\ln{(1-P_M(t))}$\;
Compute the source rate $x_s(t+1) = U_s^{'-1}\left[-\ln{\left(1-P_M(t)\right)}\right]$\;
\label{algo:source}
\end{algorithm}
\begin{algorithm}
\caption{Algorithm to  calculate prices $p_j$ and $p_e$}\label{alg:two}
\SetKwInOut{KwIn}{Input}
\SetKwInOut{KwOut}{Output}
\KwIn{$p_e(t)$, $p_j(t)$ and $x_s(t)$}
\KwOut{$p_e(t+1)$ and $p_j(t+1)$}
\textbf{@SDN (Update Source prices):}
$p_{j}(t+1) = \left[p_j(t) - \gamma_j  \left(\min C_j - \frac{\left(|C_s|(n-1)+1\right)\sigma}{d_s} - \sum_{s \neq j}x_s(t) \right)  \right]^+$\;
\textbf{@Links/Edges (Update Edge prices):}
$p_e(t+1) = \left[p_e(t) - \gamma_e \left(c_e - \sum_{s\in S(e)} x_s(t)\right)  \right]^+$\;
\textbf{@Switches: }
Sends estimated data rates of users to SDN\;
Gets the updated source and edge prices for a source\;
Marks packets with probability $m_s^e(t)$ computed using \eqref{eq:markprob}
\label{algo:switch}
\end{algorithm}
\subsection{Convergence}
It is easy to see that the primal problem \eqref{eq:modopt} is convex (considering minimization of $-U(.)$ as the objective), with linear constraints. In \cite{alghunaim2020linear,luo1993convergence}, it is shown that under the assumptions of strong convexity of the objective, the dual ascent algorithm (or descent in our case of maximization), converges linearly i.e. $|f(x) - f(x^*)|\leq \epsilon$ in $O(\log(1/\epsilon))$ iterations.
\section{Results and Discussions}
\label{sec:results}
\begin{figure}
    \centering \includegraphics[width=0.6\columnwidth]{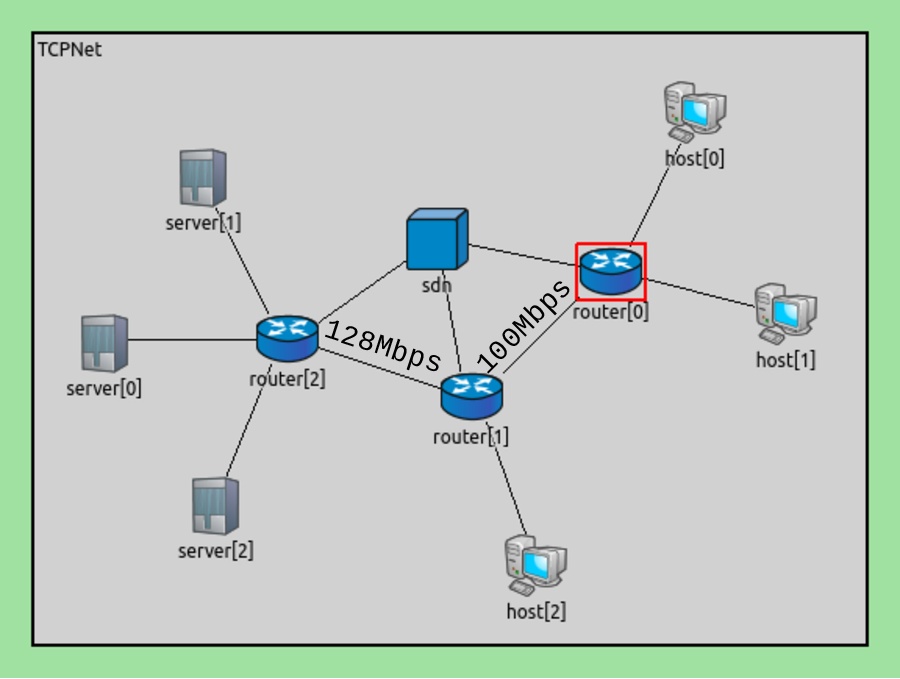}
    \caption{Simulation Setup 
 in OMNeT++. Two hosts are connected to a bottleneck link with 100 Mbps datarate while one host is connected to a 128 Mbps link which is shared by the other two as well. Host $i$ communicates with Server $i$. Other links are 1Gbps with negligible effect on delay bounds.}
    \label{fig:simusetup}
\end{figure}

In this section, we discuss three important aspects of our designed algorithm. First, we will present the steady state and transient performance of our algorithm. It is important to observe that in any delay sensitive networks, transients will eventually lead to queues which brings us to our discussion on an SDN specific digital-twin based implementation of our algorithm. Further, we discuss how our algorithm differs from the state of art schemes.
\subsection{Parameters and setup}
To evaluate the performance of the proposed TCP algorithm, we simulate a simple network as shown in \figurename \ref{fig:simusetup} using OMNeT++. The sending rates are bounded for each source $[m_s = 1 Mbps, M_s = 100 Mbps]$. We choose the utility function to be $U_s (x_s) = a_slog(1+x_s)$ with $a_s = 10^5$. Note that this is strongly concave as $U''(x_s) \leq -10^5/(1+10^6)<0$ using the fact that $x_s \geq 10^6$. We consider a maximum packet size of 1518 bytes which is the standard for Ethernet. The delay bound is set to $1$ ms which is required for satisfying eMBB slice. 

Each host is equipped with a token bucket with $\sigma = 1$ packet (1518 bytes) and as mentioned in Section \ref{sec:netback},  we regulate the bucket rate with help of our algorithm. The step sizes for switches are set as $\gamma_e = 10^{-6}$ while that for SDN is set as $\gamma_j = 10^{-7}$. We did not pursue for an optimal choice of step size, which can also be performed by using techniques in optimization theory \cite{jorge2006numerical}. All prices are set to zero at start. Price updates are performed every $5$ms. The implementation is made available at \url{https://github.com/dibbend8/tcpexperiments/tree/main/tcpSlice_Latest}.
\begin{figure*}[h]
\centering
\subfloat[]{\includegraphics[width=0.3\textwidth]{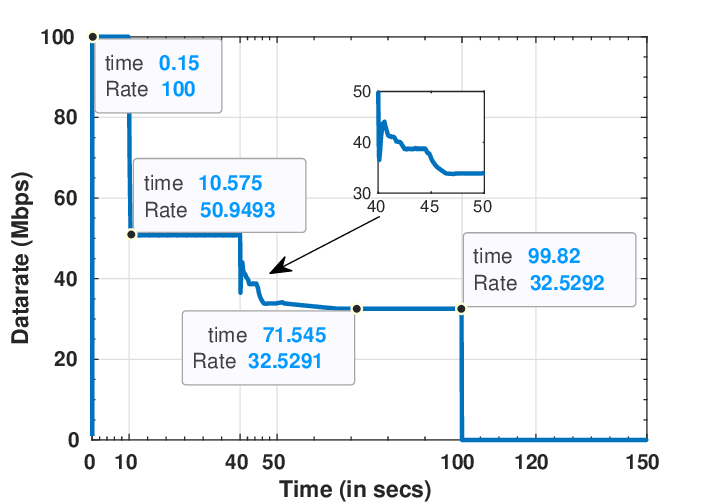}
\label{fig:host0data}}
\quad
\subfloat[]{\includegraphics[width=0.3\textwidth]{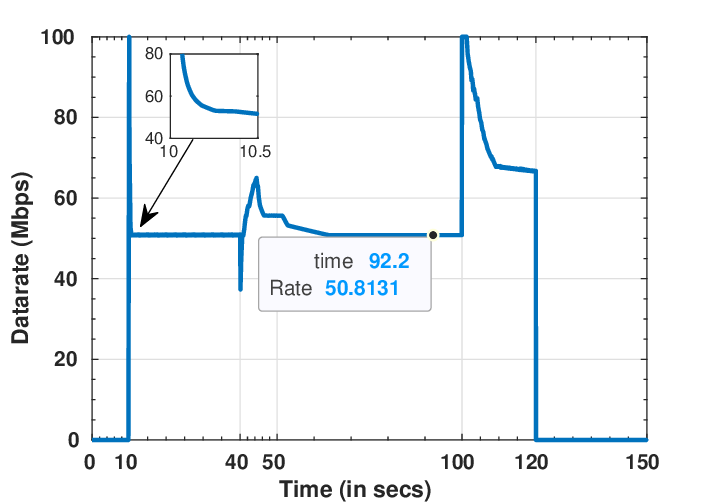}
\label{fig:host1data}}
\quad
\subfloat[]{\includegraphics[width=0.3\textwidth]{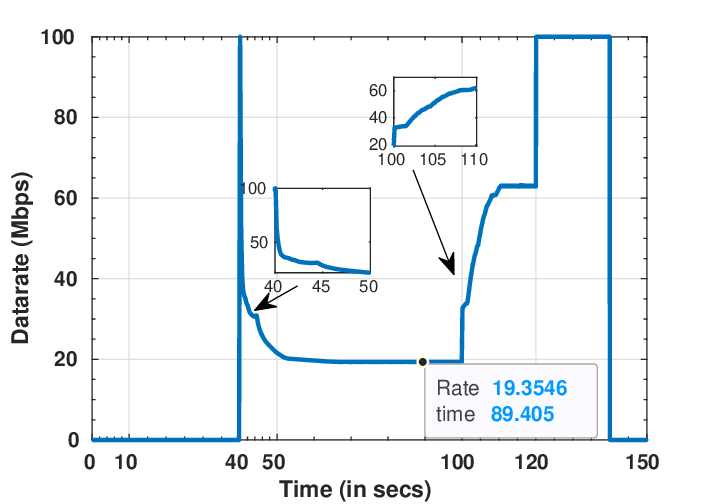}
\label{fig:host2data}}
\quad
\subfloat[]{\includegraphics[width=0.3\textwidth]{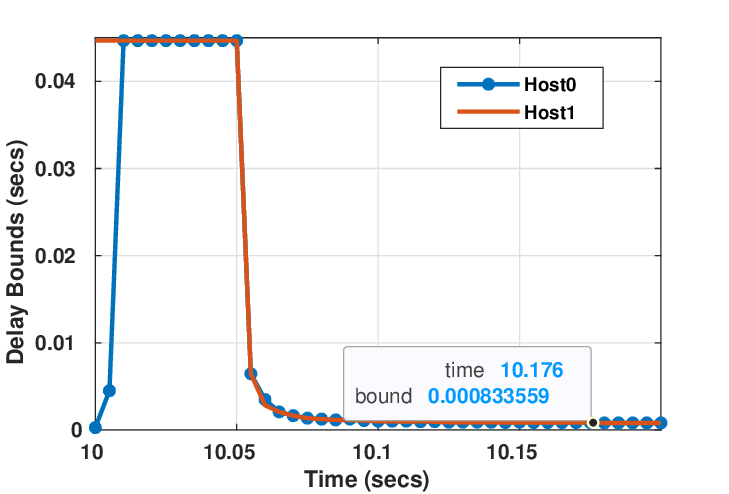}
\label{fig:bound01}}
\quad
\subfloat[]{\includegraphics[width=0.3\textwidth]{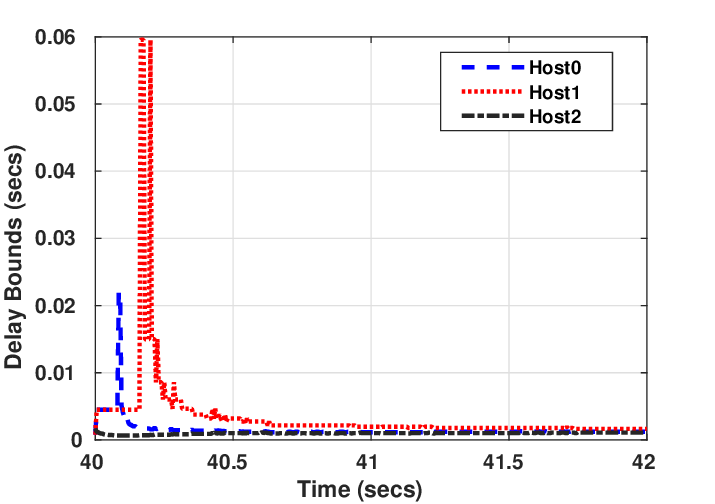}
\label{fig:bound012}}
\quad
\subfloat[]{\includegraphics[width=0.3\textwidth]{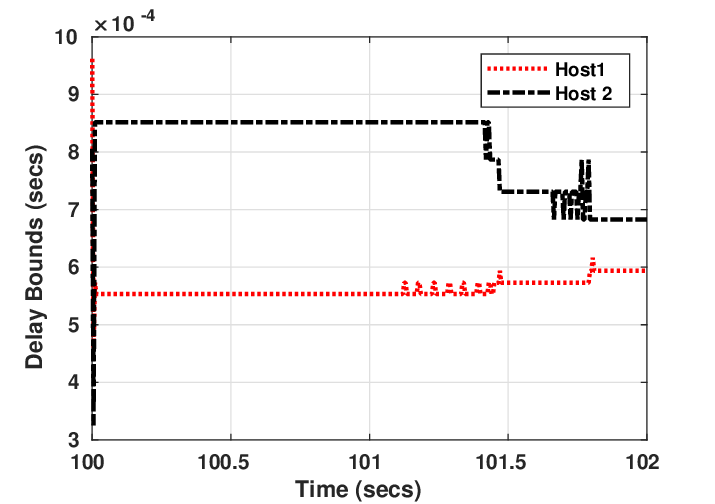}
\label{fig:bound12}}
\caption{ (a), (b), (c) Sending rates for Host 0, 1 and 2.  The hosts join and leave the network as follows: Host 0 is already in the network from the start for 10 seconds, then Host 1 joins. After 40 seconds, Host 2 joins the network. The three hosts are active till 100 seconds when Host 0 leaves the network, and then Host 1 leaves the network at 120 seconds while Host 2 leaves at 140 seconds. The figures clearly show that the steady-state rates reach the bottleneck capacity of 100 Mbps while meeting the obtained bounds. (d),(e),(f) Achieved Delay bounds. From the two figures, an important conclusion comes out: TCP slice is NOT a fair allocation scheme, rather it focuses more on delay-bound satisfaction while maximizing the network utilization.}
\label{fig:hostdata}
\end{figure*}

% \begin{figure}[h]
% \centering
% \subfloat[Delay bound convergence for hosts 0 and 1 (Host 1 joins)]{\includegraphics[0.3\textwidth]{bound01.eps}
% \label{fig:bound01}}
% \hfil
% \subfloat[Delay bound convergence for hosts 0, 1 and 2 (Host 2 joins)]{\includegraphics[0.3\textwidth]{bound012.eps}
% \label{fig:bound012}}
% \hfil
% \subfloat[Delay bound convergence for hosts 1 and 2 (Host 0 leaves)]{\includegraphics[0.3]{bound12.eps}
% \label{fig:bound12}}
% \caption{Delay bounds attained by TCP slice for different transition periods of hosts entering and leaving as described in \figurename \ref{fig:hostdata}. As seen, the bounds converge within two seconds of start of a transient phase.}
% \label{fig:bounddata}
% \end{figure}

\subsection{Steady state and Transients}
The designed experiment setup can give us a good idea about how the network behaves over time, how fast it reaches the steady state (optima in our case), and its behavior in the transient phases. \figurename \ref{fig:hostdata} shows the life cycle of hosts entering and leaving the network. Host 0 sends data at a 1Mbps rate. The switches and SDN on receiving new flow, send zero prices which leads Host 0 to boost its sending rate to the maximum. Note that we did not explicitly implement the ECN mechanism and Host 0 would spend some update cycles figuring out the price as it is implemented using probabilistic marking. Since, there are no flows from other sources, and the computed delay bound is minimal ($\sim$ 0.00023 seconds), Host 0 enjoys the full capacity for the first 10s.  

When Host 1 joins, initial prices being zero, Host 1 sends data and SDN computes the price for Host 1 where, the delay bound shoots (recall, it is $\approx 1/(100 - x_0)$ where $x_0$ is nearly 100 Mbps) resulting in a high price for Host 0, forcing it to reduce its rate (see \figurename \ref{fig:bound01}). As shown in \figurename \ref{fig:bound012}, after a few price exchanges, delay bounds reach within limits of $1$ms. In this case, for both the hosts, the 100 Mbps link is the one that produces minimum residual capacity resulting in both sharing their bandwidths around $50$ Mbps where the delay bounds are around $0.8$ms. Note that our objective is always to maximize the link utilization subject to delay constraints.

When Host 2 joins, all three compete for the link capacity while satisfying the delay constraint. Since prices are initially zero, the first price update tells it to send with 100 Mbps. In this case, the bound for Host 2 is of the order $1/(128-100)$ while for others it is $1/[128-150]^+$, which makes the delay bounds shoot for Host 0 and Host 1. It is to be noted that the delay bound calculated for Host 2 is much lower than other hosts since it traverses a single link instead of two. Since it has more room for delay bound to play with, our algorithm will penalize this host more than others to achieve the delay bounds. Thus, \textit{this algorithm does not achieve fairness as most TCP algorithms do. In its present form, our algorithm chooses bound satisfaction over fairness} (Section \ref{sec:fairvssat} for a detailed description).  Also, by looking at \figurename \ref{fig:bound012}, one might be curious regarding why Host 1's bound shoots after Host 0, whereas they must have started from the same bound value, as achieved before Host 2 joins. This is purely probabilistic. We indeed found that both started from the same delay bound value, but Host 0 got an opportunity to send its traffic a bit earlier leading to a higher delay bound calculation at Host 1. The situation could easily have been reversed and we would have observed reversed outcomes. The utilization is around 100 Mbps which is the full bandwidth of the bottleneck link. 

As shown in \figurename \ref{fig:bound12}, when Host 0 leaves the network, Host 2 is further penalized and delay bounds are adjusted between Host 1 and Host 2 so as to achieve the maximum possible throughput of 128 Mbps (see \figurename \ref{fig:host2data} and \figurename \ref{fig:host1data}). 

\subsection{SDN enabled digital-twin Implementation: Proactive vs Reactive approach}

The results show transients and their effects on delay bounds. The effects can be accounted for due to the $(1/x)$ relationship that is severely detrimental and there is no way to avoid this. Largely, TCP protocols have been reactive in nature, i.e. it reacts when something bad happens in the network whereas such a scheme would be significantly damaging when it comes to delays. The main goal of any control algorithm is to reach a steady state as soon as possible and in this case, even more so. 

When a new flow or packet arrives at an OpenFlow switch, it is forwarded to the SDN controller. The controller then directs it to the relevant application (NFV) for creating new flow tables. This allows the SDN to have visibility over new joining hosts and decide the requested level/type of service/slice mapping. By implementing a live replica of the network (a digital twin \cite{wang2022mobility}), the SDN can pre-calculate steady states before data communication begins. The SDN can compute required prices such that the host directly starts from the steady state. Nevertheless, the ECN and relevant control mechanisms should always be in place for maintenance purposes. This is a proactive approach that should lead to a much better solution for time-sensitive networks. 

\subsection{Comparison with Conventional TCP schemes}
\label{sec:fairvssat}
An important aspect of the TCP is the fairness associated with the steady state. For example, both TCP BBR \cite{cardwell2016bbr,cardwell2017bbr} and the regular TCP \cite{low1999optimization} (all variants such as RENO, VEGAS, etc., come in this category), would try to achieve a point where the bottleneck bandwidth is shared equally among the multiple flows using the bottleneck. This steady state is not ideal when there are sources with different delay constraints over the bottleneck link. Consider a situation with a shared link with a link rate of 40 Mbps and let there be applications within the uRLLC slice, one demanding a delay constraint of 1ms and the other, demanding a constraint of 0.5ms. Let us see  the resulting delay bounds if their rates are equally divided. Consider the largest packet size and burst size of 1512 bytes. Table \ref{tab:drawback} shows how a fair allocation exceeds the delay bound of a 0.5ms app, while such a requirement could be satisfied by reducing the rates of the app with a 1ms bound.

\begin{table}[h]
    \centering
    \caption{Fairness vs Delay bound satisfaction}
    \label{tab:drawback}
    \begin{tabular}{|M{0.2\columnwidth}|M{0.23\columnwidth}|c|c|}
    \hline
       $x_1$ ($d_1 = 1$ms) & $x_2$ ($d_2 = 0.5$ms) & $D_1^{NC}$ (ms) & $D_2^{NC}$ (ms)  \\
       \hline
         20     &  20   &  0.6048 & 0.6048\\
         \hline
          15    &   25   &   0.806 & 0.483\\
         \hline
    \end{tabular}
\end{table}

As more applications evolve with such diverse requirements \cite{6GConne15:online}, it would reduce the sending rate of the sources with relaxed delay constraints and increase that of others so as to maintain the delays. However, it should be noted that our approach is not fair and there may be fairer solutions in the feasible range of delay bounds. To look for fair solutions, the optimization may be reformulated to include $\sum_s \sum_{j|j\in S(e)}\sum_{e\in E(s)}|x_s - x_j|$ to be minimized in addition to maximization of $U_s$. This is separable for sources and can be implemented by each source. However, gradient computations may need to be replaced by subgradient computations due to non-differentiability of the modulus operation.

\section{Conclusion}

In this paper, we propose TCP Slice, an algorithm focusing on meeting delay constraints in a network using an optimization framework. We use network calculus to compute delay bounds and achieve a semi-distributed algorithm via dual decomposition. It involves penalizing sources for exceeding capacity and violating delays. We present the requisite technique to implement the same with help of a SDN controller and ECN feedback. We observe the steady state and transient performance of our algorithm which highlights the aspect of \textit{fairness vs bound satisfaction} in our algorithm. Furthermore, we present a proactive approach so that hosts do not encounter unnecessary delays in the transient phases. We compare our approach with existing TCP algorithms and find it to be lacking in fairness which can be rectified by changing the objective. We suggest exploring stochastic network calculus for broader applications, especially in wireless scenarios.
\label{sec:conclusion}

% or


% \appendices

% \section{}

% \section*{Acknowledgment}


\ifCLASSOPTIONcaptionsoff
  \newpage
\fi



\bibliographystyle{IEEEtran}
\bibliography{ref}

\begin{thebibliography}{10}
\providecommand{\url}[1]{#1}
\csname url@samestyle\endcsname
\providecommand{\newblock}{\relax}
\providecommand{\bibinfo}[2]{#2}
\providecommand{\BIBentrySTDinterwordspacing}{\spaceskip=0pt\relax}
\providecommand{\BIBentryALTinterwordstretchfactor}{4}
\providecommand{\BIBentryALTinterwordspacing}{\spaceskip=\fontdimen2\font plus
\BIBentryALTinterwordstretchfactor\fontdimen3\font minus
  \fontdimen4\font\relax}
\providecommand{\BIBforeignlanguage}[2]{{%
\expandafter\ifx\csname l@#1\endcsname\relax
\typeout{** WARNING: IEEEtran.bst: No hyphenation pattern has been}%
\typeout{** loaded for the language `#1'. Using the pattern for}%
\typeout{** the default language instead.}%
\else
\language=\csname l@#1\endcsname
\fi
#2}}
\providecommand{\BIBdecl}{\relax}
\BIBdecl

\bibitem{jacobson1988congestion}
V.~Jacobson, ``Congestion avoidance and control,'' \emph{ACM SIGCOMM computer
  communication review}, vol.~18, no.~4, pp. 314--329, 1988.

\bibitem{han2020framework}
L.~Han, Y.~Qu, L.~Dong, and R.~Li, ``A framework for bandwidth and latency
  guaranteed service in new ip network,'' in \emph{IEEE INFOCOM 2020-IEEE
  Conference on Computer Communications Workshops (INFOCOM WKSHPS)}.\hskip 1em
  plus 0.5em minus 0.4em\relax IEEE, 2020, pp. 85--90.

\bibitem{5GKPIs5G30:online}
``5g kpis | 5g key performance indicators,''
  \url{https://www.rfwireless-world.com/Terminology/5G-KPIs-Key-Performance-Indicators.html},
  (Accessed on 07/26/2023).

\bibitem{giordani2020toward}
M.~Giordani, M.~Polese, M.~Mezzavilla, S.~Rangan, and M.~Zorzi, ``Toward 6g
  networks: Use cases and technologies,'' \emph{IEEE Communications Magazine},
  vol.~58, no.~3, pp. 55--61, 2020.

\bibitem{6GConne15:online}
``6g – connecting a cyber-physical world - ericsson,''
  \url{https://www.ericsson.com/en/reports-and-papers/white-papers/a-research-outlook-towards-6g},
  (Accessed on 07/26/2023).

\bibitem{benzekki2016software}
K.~Benzekki, A.~El~Fergougui, and A.~Elbelrhiti~Elalaoui, ``Software-defined
  networking (sdn): a survey,'' \emph{Security and communication networks},
  vol.~9, no.~18, pp. 5803--5833, 2016.

\bibitem{mckeown2008openflow}
N.~McKeown, T.~Anderson, H.~Balakrishnan, G.~Parulkar, L.~Peterson, J.~Rexford,
  S.~Shenker, and J.~Turner, ``Openflow: enabling innovation in campus
  networks,'' \emph{ACM SIGCOMM computer communication review}, vol.~38, no.~2,
  pp. 69--74, 2008.

\bibitem{foukas2017network}
X.~Foukas, G.~Patounas, A.~Elmokashfi, and M.~K. Marina, ``Network slicing in
  5g: Survey and challenges,'' \emph{IEEE communications magazine}, vol.~55,
  no.~5, pp. 94--100, 2017.

\bibitem{jacobson1992tcp}
V.~Jacobson, R.~Braden, and D.~Borman, ``Tcp extensions for high performance,''
  Tech. Rep., 1992.

\bibitem{brakmo1994tcp}
L.~S. Brakmo, S.~W. O'malley, and L.~L. Peterson, ``Tcp vegas: New techniques
  for congestion detection and avoidance,'' in \emph{Proceedings of the
  conference on Communications architectures, protocols and applications},
  1994, pp. 24--35.

\bibitem{sally2003highspeed}
F.~Sally, ``Highspeed tcp for large congestion windows,'' \emph{RFC3649}, 2003.

\bibitem{cardwell2016bbr}
N.~Cardwell, Y.~Cheng, C.~S. Gunn, S.~H. Yeganeh, and V.~Jacobson, ``Bbr:
  Congestion-based congestion control: Measuring bottleneck bandwidth and
  round-trip propagation time,'' \emph{Queue}, vol.~14, no.~5, pp. 20--53,
  2016.

\bibitem{cardwell2017bbr}
------, ``Bbr: Congestion-based congestion control,'' \emph{Communications of
  the ACM}, vol.~60, no.~2, pp. 58--66, 2017.

\bibitem{low1999optimization}
S.~H. Low and D.~E. Lapsley, ``Optimization flow control. i. basic algorithm
  and convergence,'' \emph{IEEE/ACM Transactions on networking}, vol.~7, no.~6,
  pp. 861--874, 1999.

\bibitem{kelly1998rate}
F.~P. Kelly, A.~K. Maulloo, and D.~K.~H. Tan, ``Rate control for communication
  networks: shadow prices, proportional fairness and stability,'' \emph{Journal
  of the Operational Research society}, vol.~49, pp. 237--252, 1998.

\bibitem{low2022analytical}
S.~H. Low, \emph{Analytical methods for network congestion control}.\hskip 1em
  plus 0.5em minus 0.4em\relax Springer Nature, 2022.

\bibitem{luo1993convergence}
Z.-Q. Luo and P.~Tseng, ``On the convergence rate of dual ascent methods for
  linearly constrained convex minimization,'' \emph{Mathematics of Operations
  Research}, vol.~18, no.~4, pp. 846--867, 1993.

\bibitem{alghunaim2020linear}
S.~A. Alghunaim and A.~H. Sayed, ``Linear convergence of primal--dual gradient
  methods and their performance in distributed optimization,''
  \emph{Automatica}, vol. 117, p. 109003, 2020.

\bibitem{luo2017standardization}
J.~Luo, J.~Jin, and F.~Shan, ``Standardization of low-latency tcp with explicit
  congestion notification: A survey,'' \emph{IEEE Internet Computing}, vol.~21,
  no.~1, pp. 48--55, 2017.

\bibitem{le2001network}
J.-Y. Le~Boudec and P.~Thiran, \emph{Network calculus: a theory of
  deterministic queuing systems for the internet}.\hskip 1em plus 0.5em minus
  0.4em\relax Springer, 2001.

\bibitem{chang2000performance}
C.-S. Chang, \emph{Performance guarantees in communication networks}.\hskip 1em
  plus 0.5em minus 0.4em\relax Springer Science \& Business Media, 2000.

\bibitem{athuraliya2001rem}
S.~Athuraliya, V.~H. Li, S.~H. Low, and Q.~Yin, ``Rem: Active queue
  management,'' in \emph{Teletraffic Science and Engineering}.\hskip 1em plus
  0.5em minus 0.4em\relax Elsevier, 2001, vol.~4, pp. 817--828.

\bibitem{jorge2006numerical}
N.~Jorge and J.~W. Stephen, \emph{Numerical optimization}.\hskip 1em plus 0.5em
  minus 0.4em\relax Spinger, 2006.

\bibitem{wang2022mobility}
Z.~Wang, R.~Gupta, K.~Han, H.~Wang, A.~Ganlath, N.~Ammar, and P.~Tiwari,
  ``Mobility digital twin: Concept, architecture, case study, and future
  challenges,'' \emph{IEEE Internet of Things Journal}, vol.~9, no.~18, pp.
  17\,452--17\,467, 2022.

\end{thebibliography}
\end{document}